\documentclass[12pt,reqno]{amsart}
\numberwithin{equation}{section}
\numberwithin{figure}{section}
\usepackage[mathscr]{eucal}
\usepackage{amssymb, latexsym}
\usepackage{amsmath}
\usepackage{amsthm}
\usepackage{hyperref}
\usepackage{lpic}
\usepackage{verbatim}
\usepackage{mathtools}
\usepackage{graphicx}

\usepackage{color}

\usepackage{bbm}
\newcommand{\One}{\mathbbm{1}}

\newtheorem{thm}{Theorem}[section]
\newtheorem{lem}[thm]{Lemma}
\newtheorem{Def}{Definition}[section]
\newtheorem{prop}[thm]{Proposition}
\newtheorem{rem}[Def]{Remark}
\newtheorem{cor}[thm]{Corollary}
\newtheorem{ques}[thm]{Question}

\newcommand\bbR{{\mathbb R}}
\newcommand\bbZ{{\mathbb{Z}}}

\renewcommand\S{\Sigma}

\renewcommand\d{\partial}

\newcommand\e{\epsilon}

\renewcommand\t{\tau}
\renewcommand\th{\theta}

\newcommand{\tr}{\operatorname{tr}}

\newcommand\beq{\begin{eqnarray}}
\newcommand\eeq{\end{eqnarray}}
\newcommand\ben{\begin{enumerate}}
\newcommand\een{\end{enumerate}}
\newcommand\bit{\begin{itemize}}
\newcommand\eit{\end{itemize}}

\newcounter{mnotecount}[section]

\begin{document}

\title[On the geometry and topology of initial data sets]
{On the geometry and topology of initial data sets with horizons}

\author[L. Andersson]{Lars Andersson}
\address{Max-Planck-Institut f{\"u}r Gravitationsphysik \\
Albert-Einstein-Institut \\
Am M\"uhlenberg 1 \\
14476 Golm \\
Germany}
\email{laan@aei.mpg.de}

\author[M. Dahl]{Mattias Dahl}
\address{Institutionen f\"or Matematik \\
Kungliga Tekniska H\"ogskolan \\
100 44 Stockholm \\
Sweden} 
\email{dahl@math.kth.se}

\author[G. J. Galloway]{Gregory J. Galloway}
\address {University of Miami \\ 
Department of Mathematics \\ 
Coral Gables, FL 33124 \\ 
U.S.A.}
\email {galloway@math.miami.edu}

\author[D. Pollack]{Daniel Pollack}
\address {University of Washington \\ 
Department of Mathematics \\ 
Seattle, WA 98195 \\U.S.A.} 
\email {pollack@math.washington.edu}

\begin{abstract}
We study the relationship between initial data sets with horizons and
the existence of metrics of positive scalar curvature. We define a
Cauchy Domain of Outer Communications (CDOC) to be an asymptotically
flat initial set $(M, g, K)$ such that the boundary $\partial M$ of
$M$ is a collection of Marginally Outer (or Inner) Trapped Surfaces
(MOTSs and/or MITSs) and such that $M\setminus \partial M$ contains no
MOTSs or MITSs.  This definition is meant to capture, on the level of
the initial data sets, the well known notion of the domain of outer
communications (DOC) as the region of spacetime outside of all the
black holes (and white holes). Our main theorem establishes that in
dimensions $3\leq n \leq 7$, a CDOC which satisfies the dominant
energy condition and has a strictly stable boundary has a positive
scalar curvature metric which smoothly compactifies the asymptotically
flat end and is a Riemannian product metric near the boundary where
the cross sectional metric is conformal to a small perturbation of the
initial metric on the boundary $\partial M$ induced by $g$. This
result may be viewed as a generalization of Galloway and Schoen's
higher dimensional black hole topology theorem \cite{GS06} to the
exterior of the horizon. We also show how this result leads to a
number of topological restrictions on the CDOC, which allows one to
also view this as an extension of the initial data topological
censorship theorem, established in \cite{EGP13} in dimension $n=3$, to
higher dimensions.
\end{abstract}


\maketitle

\section{Introduction}

One of the interesting features of general relativity is that it does
not a priori impose any restrictions on the topology of space. In
fact, as was shown in \cite{IMP}, given a compact manifold $M$ of
arbitrary topology and a point $p\in M$, there always exists an
asymptotically flat solution to the vacuum Einstein constraint
equations on $M\setminus \{p\}$. However, according to the principle
of topological censorship, the topology of the domain of outer
communications (DOC), that is the region outside of all black holes
(and white holes), should, in a certain sense, be simple. The
rationale for this is roughly as follows. Known results
\cite{Gannon,Lee} suggest that nontrivial topology tends to induce
gravitational collapse. In the standard collapse scenario, based on
the weak cosmic censorship conjecture, the process of gravitational
collapse leads to the formation of an event horizon which shields the
singularities from view. As a result, according to the viewpoint of
topological censorship, the nontrivial topology gets hidden behind the
event horizon, and hence the DOC should have simple topology. There
have been a number of results supporting this point of view, the most
basic of which establishes the simple connectivity of the DOC in
asymptotically flat spacetimes obeying suitable energy and causality
conditions \cite{FSW,Gdoc}. However, all the results alluded to here
are spacetime results, that is, they involve conditions that are
essentially global in time.

In \cite{EGP13} a result on topological censorship was obtained at the
pure initial data level for asymptotically flat initial data sets,
thereby circumventing difficult questions of global evolution; see
\cite[Theorem~5.1]{EGP13}. This result, which establishes, under
appropriate conditions, the topological simplicity of $3$-dimensional
asymptotically flat initial data sets with horizons, relies heavily on
deep results in low dimensional topology, in particular the resolution
of the Poincar\'e and the geometrization conjectures. The aim of the
present paper is to obtain some results of a similar spirit, but in
higher dimensions. While similar in spirit, the methods we employ here
are entirely different. Given an asymptotically flat initial data set
which satisfies the dominant energy condition and contains an inner
horizon (marginally outer trapped surface), we use Jang's equation
(\cite{SY2, AEM}), and other techniques, to deform the metric to one
of positive scalar curvature on the manifold obtained by compactifying
the end, such that the metric has a special structure near the
horizon. As we shall discuss, one can then use known obstructions to
the existence of such positive scalar curvature metrics to obtain
restrictions on the topology of the original initial data
manifold. Related approaches to the topology of asymptotically flat
initial data sets {\it without} horizons have been considered in
\cite{SYPRL, SchWItt}. Here we must overcome a number of difficulties
due to the presence of a horizon.

An initial data set $(M,g,K)$ for Einstein's equations consists of an 
$n$-dimensional manifold $M$, a Riemannian metric $g$ on $M$, and a 
symmetric $2$-tensor $K$ on $M$. The energy density $\mu$ and the 
momentum density $J$ of $(M,g,K)$ are computed through
\[
2\mu = R^g - |K|_g^2 + (\tr^g K)^2
\quad \text{and} \quad
J = \operatorname{div}^g K - d(\tr^g K).
\]
The initial data set $(M,g,K)$ satisfies the dominant energy condition
if 
\[
\mu \geq |J| .
\]

Let $\Sigma$ be a compact $2$-sided hypersurface in an initial data
set $(M,g,K)$. Then $\S$ admits a smooth unit normal field $\nu$ in
$M$. By convention, refer to such a choice as outward pointing. Then
the outgoing and ingoing null expansion scalars $\th_{\pm}$ are
defined in terms of the initial data as $\th_{\pm} = P \pm H$, where
$P = \tr_\Sigma K$ is the partial trace of $K$ along $\S$ and $H$ is
the mean curvature of $\S$, which, by our conventions, is the
divergence of $\nu$ along $\S$. We call $\Sigma$ outer 
trapped if $\theta_+ < 0$ on $\Sigma$, 
while if $\theta_- < 0$, $\Sigma$ is inner
trapped. We call $\Sigma$ a marginally outer trapped surface (MOTS)
if $\theta_+ = 0$, while if $\theta_- = 0$, we call $\Sigma$ a
marginally inner trapped surface (MITS). The distinction between MOTS
and MITS is only meaningful when a choice, natural or otherwise, has
been made between the notions of ``outside'' and ``inside''.

Galloway and Schoen have proved the following extension of Hawking's 
black hole topology theorem to higher dimensions. 

\begin{thm}[\cite{GS06}] \label{GS-thm}
Let $(M, g, K)$ be an $n$-dimensional, $n \geq 3$ initial data set
satisfying the dominant energy condition. If $\Sigma$ is a stable
MOTS, in particular if $\S$ is outermost, then, apart from
certain exceptional circumstances, $\Sigma$ is of positive Yamabe
type.
\end{thm}

In the situation of this theorem, let $h$ denote the induced metric on
$\Sigma$. The conclusion is then that $h$ is conformal to a metric of
positive scalar curvature. The ``exceptional circumstances'' can be ruled
out in various ways \cite{GS06, G08}, in particular if $\S$ is assumed
to be {\it strictly} stable. See \cite{AEM} and references therein for
the notion of MOTS stability.

The main result of the present paper is an extension of the above
theorem for asymptotically flat initial data sets, stating that the
positive scalar curvature metric on $\Sigma$ can be extended to the
one-point compactification of the asymptotically flat manifold
consisting of the exterior of $\Sigma$. 

\begin{Def} \label{def:CDOC} 
A Cauchy Domain of Outer Communications (CDOC) is an asymptotically
flat initial set $(M, g, K)$ such that the boundary of $M$, $\partial
M$ is a collection of MOTSs and/or MITSs and such that
$M\setminus \partial M$ contains no MOTSs or MITSs.
\end{Def}

The precise form of asymptotic flatness we require is given in
Definition~\ref{def:AF}.

Definition \ref{def:CDOC} is meant to capture, strictly on the level
of initial data sets, the well known notion of the DOC as the region
of spacetime outside of all the black holes (and white holes). More
precisely, it is meant to model an asymptotically flat (partial)
Cauchy surface within the DOC, with boundary on the event horizon (in
the equilibrium case) or perhaps somewhat inside the event horizon (in
the dynamic case). Our main theorem, as noted above, establishes the
existence of a particular type of positive scalar curvature metric on
a CDOC.
 
\begin{thm} \label{mainthm}
Let $(M, g, K)$ be an $n$-dimensional, $3 \leq n \leq 7$, CDOC whose
boundary $\S = \partial M$ is connected and is a strictly stable
MOTS. Suppose further that the initial data $(g,K)$ extends to a
slightly larger manifold $N$ (which contains $M$ and a collar
neighborhood of $\S$) such that the dominant energy condition (DEC)
holds on $N$, $\mu \ge |J|$.

Let $\check{M}$ denote $M$ with the asymptotically flat end 
compactified by a point. Let $h$ denote the metric on $\partial M$ 
induced from $g$. Then $\check{M}$ admits a positive scalar 
curvature metric $\check{g}$ 
\begin{enumerate}
\item
whose induced metric on the boundary $\partial M$ is conformal to a small
perturbation of $h$, and
\item 
is a Riemannian product metric in a collar neighborhood of $\partial M$.
\end{enumerate}
\end{thm}

\begin{rem} 
Theorem \ref{mainthm} is stated for simplicity for a CDOC with a
single outermost MOTS $\Sigma$. The analogous statement for the case
where $\Sigma$ is a collection of outermost MOTSs and MITSs can easily
be proved along the same lines. 
\end{rem} 

\begin{rem} 
The proof in the $n=3$ case of Theorem~\ref{mainthm} requires a
modification from the general case when $n>3$. This is addressed in
Remarks~\ref{rem:n=3a} and \ref{rem:n=3}. The restriction that $n\leq
7$ in Theorem~\ref{mainthm} is the result of our use of existence
results for smooth solutions of Jang's equation (see in particular
Theorems~\ref{thmstrictDEC} and \ref{thmexist} below). This
restriction is closely related to the partial regularity imposed in
higher dimensions by the existence of the Simons cone, a singular area
minimizing hypersurface in ${\bbR}^8$.
\end{rem}

The existence of a positive scalar curvature metric on the
compactification of $M$ (and its double, which follows immediately
from the product structure near the boundary) gives restrictions on
the topology of $M$. In Section~\ref{sec_PSC} we discuss such
restrictions in more detail.

The black ring spacetime of Emparan and Reall \cite{ER}, which is an
asymptotically flat, stationary solution to the vacuum Einstein
equations, illustrates certain features of our results. Let $M$ be the
closure of a Cauchy surface for the domain of outer communications of
the black ring. The boundary of $M$ coincides with the bifurcate
horizon, which has topology $S^2 \times S^1$.  
Further, as shown in
\cite{Chru} (see also \cite{Kunduri}), the compactification $\check{M}$ of $M$ has
topology $S^2 \times D^2$. 
This is consistent with
Theorem~\ref{mainthm}, as well as standard results on topological
censorship, which require the domain of outer communications to be
simply connected. Note that, while $\pi_1(\check{M})$ is trivial,
$H_2(\check{M},\bbZ) = \bbZ \ne 0$.

An essential part of the argument is to show that we can specialize to 
the case in which dominant energy condition holds strictly, $\mu > |J|$. 
This involves a perturbation of the initial data, as discussed in 
Section~\ref{sec_DEF}. It is here that we need the assumption that 
$\S$ is {\em strictly} stable. 

This paper is a contribution to the long history of results tying the
existence of metrics of positive scalar curvature to the analysis of
initial data sets in general relativity. One of the earliest and most
important examples of this is the transition from Schoen and Yau's
work on topological obstructions to positive scalar curvature metrics
\cite{SYscalar, SYscalar2} to their proof, using minimal
hypersurfaces, of the positive mass theorem \cite{SY1, SY2}. The
results here, like Theorem \ref{GS-thm}, make strong ties between the
dominant energy condition, the presence of marginally trapped surfaces
and metrics of positive scalar curvature.

\subsection*{Acknowledgements}

The authors wish to thank Michael Eichmair, Lan-Hsuan Huang and Anna
Sakovich for many helpful comments on the work in this paper.

\section{Deforming to strict dominant energy condition}
\label{sec_DEF}

We start by introducing an appropriate notion of asymptotically flat
initial data. We shall use the conventions of \cite{EHLS} for function
spaces, which agree with the conventions of \cite{bartnik:mass}. All
sections of bundles are assumed to be smooth unless otherwise
stated. The following definition is an adaption of
\cite[Definition~3]{EHLS} to our situation.

\begin{Def} \label{def:AF} 
Let $(M,g,K)$ be an initial data set of dimension $n \geq 3$. 
Let $k$ be an integer, $k \geq 3$. Further, let $p > n$, 
$q \in ((n-2)/2, n-2)$, $q_0 > 0$, $\alpha \in (0,1-n/p]$. We say that 
$(M,g,K)$ is \emph{asymptotically flat} (of type $k,p,q,q_0,\alpha$) 
if there is a compact set $Y \subset M$ and a $C^{k+1,\alpha}$ 
diffeomorphism identifying $M\setminus Y$ with $\bbR^n \setminus B$ for 
some closed ball $B \subset \bbR^n$, for which 
\begin{equation}\label{eq:gK-reg}
(g - \delta , K) \in 
W^{k,p}_{-q}(\bbR^n \setminus B) \times W^{k-1,p}_{-q-1}(\bbR^n\setminus B)
\end{equation}
where $\delta$ is the standard flat metric on $\bbR^n$, and 
\begin{equation}\label{eq:muJ-reg}
(\mu, J) \in C^{k-2,\alpha}_{-n-q_0} \,.
\end{equation} 
\end{Def} 

In order to avoid certain technical problems our definition of asymptotic 
flatness differs from that of \cite{EHLS} by assuming higher regularity. 
Note that our assumptions imply pointwise estimates for two derivatives 
of $g - \delta$ which is not valid under the assumptions of \cite{EHLS}. 
Note also that the condition \eqref{eq:muJ-reg} which is adapted from 
\cite[Definition~3]{EHLS} implies additional fall-off for $(\mu,J)$ 
over that implied by \eqref{eq:gK-reg}. 

We shall make use of the weighted Sobolev and H\"older spaces in the setting 
of manifolds with boundary. Definition \ref{def:AF} extends immediately 
to this situation. 

The aim of this section is to establish the following perturbation result. 

\begin{thm} \label{thmstrictDEC} 
Let $(M, g, K)$ be an 
initial data set of dimension $n$, $4 \le n \le 7$, 
which is asymptotically flat in the sense of Definition~\ref{def:AF} and 
such that $M$ is a manifold with boundary, whose boundary $\S = \partial M$ 
is a connected strictly stable MOTS. Suppose further:
\begin{enumerate}
\item[(i)] 
The initial data $(g,K)$ extends to a slightly larger manifold $N$ (which 
contains $M$ and a collar neighborhood of $\S$) such that the dominant 
energy condition (DEC) holds on $N$, $\mu \ge |J|$.
\item[(ii)] There are no MOTSs or MITSs in $M \setminus \S$.
\end{enumerate}
Then for $\epsilon > 0$ there is an asymptotically flat initial data 
$(\hat g, \hat K)$ such that
\[
\| \hat{g} - g \|_{W^{k,p}_{-q}} + \| \hat{K} - K\|_{W^{k-1,p}_{-1-q}} < \epsilon 
\]
and a manifold $\hat M \subset N$ diffeomorphic to $M$, with MOTS
boundary $\hat \S = \d \hat M$, which is a small (in $C^{2,\alpha}$)
perturbation of $\S$, such that the following statements hold.
\begin{enumerate}
\item \label{point:jang1} 
The dominant energy condition holds strictly on $(\hat M, \hat g, \hat K)$, 
that is
\[
\hat \mu > |\hat J|. 
\]
\item \label{point:jang2} 
There exists a smooth solution to Jang's equation 
$u : \hat M \setminus \hat \S \to \bbR$, such that 
\begin{equation} \label{eq:jang-decay} 
u \in W^{k+1,p}_{1-q}
\end{equation}
and $u \to \infty$ on approach to $\hat \S$. 
\end{enumerate}
\end{thm}

\begin{rem}\label{rem:n=3a}
In the statement of Theorem~\ref{thmstrictDEC} we have excluded the
case $n=3$. The reason for this is that in the proof we are making use
of the density theorem \cite[Theorem~22]{EHLS}, which yields deformed
data $(\hat{g}, \hat{K}$) satisfying the strict dominant energy
condition, and with the same asymptotic behavior as $(g,K)$, in
particular $\tr^{\hat{g}} \hat{K} \in W^{k-1,p}_{-1-q}$, which in
case $n=3$ is in general incompatible with having a bounded solution
(near infinity) to Jang's equation.
This problem does not arise for $n \ge 4$ in which case 
$\tr^{\hat{g}} \hat{K} = O(|x|^{-\gamma})$ for some $\gamma > 2$. In 
Theorem~\ref{thmexist} below, which does not rely on the just mentioned 
density theorem, we have avoided this technical point by including the 
additional assumption \eqref{eq:trgK-decay}. In Remark~\ref{rem:n=3} 
below we describe the modifications necessary to prove 
Theorem~\ref{mainthm} in the case $n = 3$.
\end{rem}

The proof involves several elements. We begin with some comments about 
Jang's equation. Schoen and Yau \cite{SY2} studied in detail the existence 
and regularity of solutions to Jang's equation in their proof of the 
positive mass theorem in the general (not time-symmetric) case. They 
interpreted Jang's equation geometrically as a prescribed mean curvature 
equation, and discovered that the only possible obstruction to global 
existence are MOTSs in the initial data, where, in fact, the solution 
may have cylindrical blow-ups.

Given an initial data set $(M,g,K)$, consider graphs of functions 
$u: M \to \bbR$ in the initial data set $(\bar M, \bar g, \bar K)$ of 
one dimension higher, where $\bar M = M \times \bbR$, $\bar g = g + dt^2$, 
and $\bar K$ is the pullback of $K$ to $\bar M$ by the projection to $M$. 
Jang's equation may then be written as
\[
H(u) -\tr\bar K(u) = 0,
\]
where $H(u)$ is the mean curvature of $\operatorname{graph}(u)$, with
respect to the downward pointing normal,
in $(\bar M, \bar g)$ and $\tr\bar K(u)$ is the partial trace of 
$\bar K$ over the tangent spaces of $\operatorname{graph}(u)$.

The fundamental existence result of Schoen and Yau
\cite[Proposition~4]{SY2} for Jang's equation may now be applied. We
also rely on the work of Metzger \cite{Metzger} to allow for an
interior barrier, and the regularity theory (up to dimension 7) of
Eichmair \cite{Eichmair1,Eichmair2} (see also
\cite{EM,eichmair:2013CMaPh.319..575E}). This together yields the
following existence result for Jang's equation in our setting.

\begin{thm} \label{thmexist} 
Let $(M,g,K)$ be an initial data set of dimension $n$, $3 \leq n \leq 7$ 
which is asymptotically flat in the sense of definition \ref{def:AF}. 
In case $n=3$, we require that $\tr^g K$ satisfies the additional decay 
condition
\begin{equation} \label{eq:trgK-decay} 
\tr^g K = O(|x|^{-\gamma})
\end{equation}
for some $\gamma > 2$. Further, we assume that $M$ is a manifold with 
boundary, whose boundary $\S = \d M$ is a compact connected outer trapped 
surface.

Then there exist open pairwise disjoint sets $\Omega$, $\Omega_+$ and
$\Omega_-$, with $\Omega$ containing a neighborhood of infinity, and
an extended-real valued function $u$ whose domain includes the union
$\Omega\cup \Omega_+ \cup \Omega_-$ (and is realized as a limit of
solutions to the capillarity regularized Jang equation \eqref{eqjanqreg})
such that
\begin{enumerate}
\item
$M = \overline \Omega \cup \overline \Omega_+ \cup \overline \Omega_-$.
\item 
$u = +\infty$ on $\Omega_+$, where $\Omega_+$ contains a neighborhood of 
$\S$, and $u = -\infty$ on $\Omega_-$. 
\item 
Each boundary component $\S_a^+$ of $\Omega_+$ is a MOTS (except for $\S$), 
and each 
boundary component $\S_b^-$ of $\Omega_-$ is a MITS. (Here ``outside'' is 
determined by the outward normal to these open sets.)
\item 
$u: \Omega \to (-\infty,\infty)$ is a smooth solution to Jang's equation 
such that $u(x) \to 0$ as $x \to \infty$, $u(x) \to +\infty$ as 
$x \to \d \Omega_+$, and $u(x) \to - \infty$ as $x \to \d \Omega_-$. 
The boundary components of $\Omega$ are smooth and form a subcollection 
of the MOTSs $\S_a^+$ and MITSs $\S_b^-$ in point (3). 
\end{enumerate}
\end{thm} 
\begin{rem}\label{remproof} 
To prove Theorem~\ref{thmexist} one considers the capillarity regularized Jang 
equation
\begin{equation} \label{eqjanqreg}
H(u_{\tau}) - {\tr}\bar K(u_{\t}) = \t u_{\t} ,
\end{equation}
and studies the limit as $\tau \to 0$. This regularized equation satisfies 
an a priori height estimate that allows one to construct a smooth global 
solution $u_{\tau}$ on $M \setminus \S$ such that $u_{\tau} \to 0$ on the 
asymptotically flat end (uniformly in $\t$), and $u_{\tau} \to \infty$ as 
$\t \to 0$ on a fixed neighborhood of $\S$; see 
\cite{SY2, AM2, Eichmair1, Metzger, EM}. To get smooth convergence up 
to dimension~7, one applies the method of regularity introduced in the 
study of MOTSs by Eichmair \cite{Eichmair1}, based on the $C$-minimizing 
property. By a calibration argument, the graphs 
$G_{\t} =\operatorname{graph}(u_{\t})$ obey the $C$-minimizing property 
(this is true in general for graphs of bounded mean curvature). By the 
compactness and regularity theory of $C$-minimizers as described in 
\cite[Appendix~A]{Eichmair1}, a subsequence of these graphs converges to 
a smooth hypersurface $G \subset \bar M$, consisting of cylindrical 
components (which occur at the intersection of $\d \Omega_+$ and 
$\d \Omega_-$, where $u$ is not defined) and graphical components which are also $C$-minimizing. 
It is this hypersurface $G$ that determines the open sets $\Omega$, 
$\Omega_+$ and $\Omega_-$. Considering translations of the graphical 
components of $G$ gives rise to further cylinders obeying the 
$C$-minimizing property. The collections of MOTSs $\Sigma_a^+$ and 
MITSs $\Sigma_b^-$ in part (3) of Theorem~\ref{thmexist} arise from 
the intersection of all these cylinders with $M = M \times \{0\}$
in $\bar M$. Of further importance to us, as observed in 
\cite{Eichmair1}, the $C$-minimizing property of these cylinders 
descends to the collection of MOTSs $\Sigma_a^+$ and MITSs $\Sigma_b^-$.
\end{rem} 

\subsection{Proof of Theorem~\ref{thmstrictDEC}} 

Since $\S$ is strictly stable in $(M,g,K)$, assumption (i) of the theorem enables us to construct an enlarged manifold $M' = M \cup V \subset N$, where $V$ is an exterior collar $V \approx [0, \e]
\times \S$ attached to $\Sigma = \partial M$, such that
$\S_t = \{t\} \times \S$ is outer trapped for all $t \in (0, \e]$
and $\S_0 = \S$, see Figure~\ref{figure1}. For more details, see the
discussion following Definition~3.1 in \cite{AEM}. Then $(M', g, K)$
is an asymptotically flat initial data set with boundary
$\d M' = \S' \coloneqq \S_{\e}$. On the asymptotically flat end we let
$\S(r)$ denote the radial sphere $|x| = r$.

\begin{figure}[h]
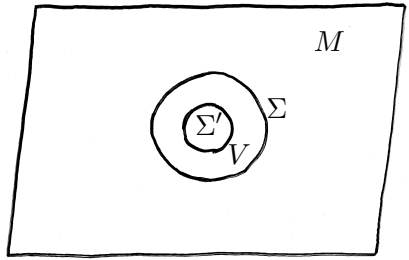

\begin{lpic}[]{FIG1(0.20)}
\lbl[]{230,160;{\small $M$}}
\lbl[]{170,85;{\small $V$}}
\lbl[]{195,115;{\small $\Sigma$}}
\lbl[]{150,105;{\small $\Sigma'$}}
\end{lpic}
\caption{$M' = M \cup V$.}
\label{figure1}
\end{figure}

\begin{lem} \label{perturb} 
With the assumptions of Theorem \ref{thmstrictDEC}, there exists a 
sequence of initial data sets $(M', g_i, K_i)$ such that
$(g_i, K_i)$ converges 
to $(g, K)$ in $W^{k, p}_{-q} \times W^{k-1, p}_{-1-q}$ and such that the 
following holds. 
\begin{enumerate}\renewcommand{\theenumi}{(\arabic{enumi})}
\renewcommand{\labelenumi}{\theenumi}
\item \label{part1}
The dominant energy condition holds strictly, $\mu_i > |J_i|_{g_i}$.
\item \label{part2}
For each data set $(M', g_i, K_i)$, $\S_t$ is outer trapped for all 
$\frac{1}{i} \le t \le \e$.
\item \label{part3}
There exists $r_0 >0$ such that for all $r \ge r_0$, $\S(r)$ is inner 
trapped ($\th_- < 0$) and outer untrapped ($\th_+ > 0$) with respect to 
each $(g_i, K_i)$.
\end{enumerate}
\end{lem}
\begin{proof} 
The proof of Lemma~\ref{perturb} is based on \cite[Theorem~22]{EHLS}. 
By considering the double of $M'$, one sees that $M'$ can be compactly 
``filled-in'' beyond its boundary to obtain a complete manifold $N'$ 
without boundary. Extend the data $(g,K)$ arbitrarily, but smoothly 
to $N'$. Then $(N',g,K)$ is an asymptotically flat manifold such that 
dominant energy condition holds on $M' \subset N'$. Thus, by 
\cite[Theorem~22]{EHLS}\footnote{For our situation we are actually 
applying a small refinement of \cite[Theorem~22]{EHLS}, whereby one does not require the dominant energy condition to hold everywhere. Essentially, the construction in \cite{EHLS} gives strictness at points where the dominant energy condition holds in the original data. We thank Lan-Hsuan Huang \cite{Huang} for clarification on this point.}
and the remark following the statement of \cite[Theorem~18]{EHLS} 
concerning higher regularity, there exists a sequence of asymptotically 
flat initial data sets $(M', g_i, K_i)$ satisfying
\[
\|g-g_i\|_{W^{k, p}_{-q}} \leq \frac1{i} 
\qquad \text{and} \qquad 
\|K-K_{i}\|_{W^{k-1, p}_{-1-q} } \leq \frac1{i}
\]
such that part \ref{part1} of Lemma~\ref{perturb} holds. Moreover, as 
follows from \cite[Equation (40)]{EHLS}, for each $i$, $(g_i,K_i)$ can 
be made sufficiently $C^2$-close to $(g,K)$ on the collar $V$ so that 
part \ref{part2} holds.

Now consider the null mean curvatures 
$\theta_{\pm}(r)= \tr_{\Sigma(r)}K \pm H$ 
(resp., $\theta_{i\pm}(r)) $) of the coordinate spheres $\S(r)$ in the 
initial data set $(g,K)$ (resp., $(g_i,K_i)$). Since the null mean 
curvatures $\theta_{\pm}(r)$ are polynomials in $g$ and its first 
derivatives, and $K$ (and similarly for ${\theta_i}_{\pm}(r)$ with respect 
to $g_i$ and $K_i$) the weighted Sobolev embedding 
$W^{k, p}_{-q} \subset C^{m,\alpha}_{-q}$ provided 
$m+\alpha < k - \frac{n}{p}$, implies that
\[
\|{\theta_i}_{\pm}(r) - \theta_{\pm}(r)\|_{ C^{0,\alpha}_{-q-1}} 
\leq \frac{C}{i}
\]
for a constant $C$ independent of $i$. This implies that 
\[
|{\theta_i}_{\pm}(r) - \theta_{\pm}(r)| = O(r^{-q-1}) .
\]
Since the mean curvature of large spheres falls off linearly with the 
radius this implies that part \ref{part3} of Lemma~\ref{perturb} holds. 
This concludes the proof of the lemma.
\end{proof} 

We now apply Theorem~\ref{thmexist} to each initial data set 
$(M', g_i, K_i)$ guaranteed by Lemma~\ref{perturb}. Thus, for each $i$ 
there exist open sets $\Omega_i$, 
$\Omega_{i+}$ and $\Omega_{i-}$ and
an extended-real valued function $u_i$ as in the theorem. In particular, 
$u_i : \Omega_i \to \bbR$ is a smooth solution to Jang's equation. Let 
$\Omega_i^{ext}$ be the component of $\Omega_i$ containing the 
asymptotically flat end. We are primarily interested in the smooth 
solutions $u_i : \Omega_i^{ext} \to \bbR$. The boundary 
$S_i \coloneqq \d \Omega_i^{ext}$ consists of MOTSs $S_{i,a}^+$ and 
MITSs $S_{i,b}^-$. (Consistent with Theorem~\ref{thmexist}, the
``outside'' is determined by the normal pointing into $\Omega_i^{ext}$.)
Here we use indices $a,b$ to enumerate the MOTS and 
MITS components of $S_i$. We have $u_i \to +\infty$ on 
approach to the MOTSs $S_{i,a}^+$ and $u_i \to -\infty$ on approach to 
the MITSs $S_{i,b}^-$.

Let $\Omega_{i+}'$ be the component of $\Omega_{i+}$ containing $\S'$. 
From Theorem~\ref{thmexist}, $\Omega_{i+}' \ne \emptyset$ for all $i$. 
In fact, as in \cite{Metzger}, the maximum principle implies that 
$\Omega_{i+}' \supset V_i \coloneqq \cup_{t \in [\frac1{i}, \e]} \S_t$. This implies, 
in particular, that $S_i \ne \emptyset$ for all $i$. Moreover, by part 
\ref{part3} of Lemma~\ref{perturb} and the maximum principle, 
$S_i \subset M'(r_0)$, where $M'(r_0) \subset M'$ is the compact region 
bounded by $\S(r_0)$. By the convergence of the data $(g_i, K_i)$ to 
$(g,K)$ on $M'(r_0)$, the sequence $S_i$ obeys a uniform $C$-minimizing 
property, see Remark~\ref{remproof}. Hence, by the compactness theory 
presented in \cite{Eichmair1, Eichmair2} (which provides area bounds, 
curvature bounds and injectivity bounds), by passing to a subsequence if 
necessary, the sequence $S_i$ converges in $C^{2,\alpha}$
to $S$ which is a combination of 
MOTSs and MITSs in $(M,g,K)$. Note that no component of $S$ enters the 
collar region exterior to $M$ in $M'$ since $S_i \cap V_i = \emptyset$ 
for all $i$.

We shall now prove that $S=\Sigma$. By the above there is a unique
smallest collection $\mathring{S_i}$ of components of $S_i$
surrounding $\Sigma'$ in the sense that $\mathring{S_i}$ separates
$\Sigma'$ from infinity, and hence a unique smallest collection
$\mathring{S}$ of components of $S$ surrounding $\Sigma$. Since
$\mathring{S}$ consists of MOTSs and MITSs, and since our assumptions
exclude any MOTSs or MITSs in the exterior of $\Sigma$, it follows
that $S = \mathring{S}$ and that each component of $\mathring{S}$
meets $\Sigma$ at some point.

\begin{figure}[h]
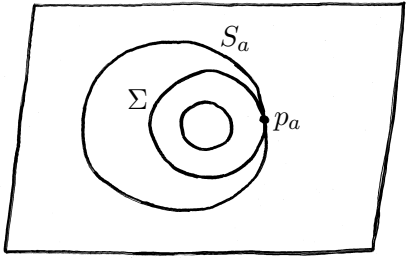

\begin{lpic}[]{FIG2(0.20)}
\lbl[]{105,120;{\small $\Sigma$}}
\lbl[]{205,105;{\small $p_a$}}
\lbl[]{170,160;{\small $S_a$}}
\end{lpic}
\caption{$S_a$ meets $\Sigma$ at $p_a$.}
\label{figure2}
\end{figure}

Let now $S_a$ be one of the components of $S$, which by the above must 
meet $\Sigma$ at some point $p_a$, see 
Figure~\ref{figure2}. Suppose that $S_a$ is a MOTS. Since the 
outward pointing normal of $S_a$ must agree with that of $\Sigma$, the 
maximum principle implies that in this case $S_a = \Sigma$ and we are done. 
It remains to consider the case when $S_a$ is a MITS. By construction, $S_a$ 
is the limit of a sequence $S^-_{i,a}$ of components of the boundary of 
$\Omega_i^{ext}$,  each of which is separated from $\Sigma_{1/i}$ by a part of 
$\Omega_{i-}$. Pulling back slightly from the limit, i.e.\ for very large $i$ and for values of $\tau$ sufficiently small, 
in a small neighborhood $U_a$ of $p_a$, the capillarity regularized Jang graph, ${\rm graph}\, u_{i,\tau}$, and its vertical translates come uniformly close to the vertical cylinders over the MITS $S^-_{i,a}$ and a MITS component $T_i$ of $\Omega_{i-}$ inside $S^-_{i,a}$, which converges to $\Sigma$ (see Figure~\ref{figure3}).   For $i$ very large and values of $\tau$ sufficiently small,  this can be seen to lead to a violation of the uniform (in both $i$ and 
$\tau$) $C$-minimizing property of these capillarity regularized Jang graphs  \cite{Eichmair1, Eichmair2} (for example by gluing in a tube.)
Thus, the case that $S_a$ is a MITS meeting $\Sigma$ is precluded by this $C$-minimizing property.  Hence we find that $S = \Sigma$.
(We remark that instead of considering the limits 
$\tau \searrow 0$ and $i \nearrow \infty$ separately, a simultaneous limit in $\tau, i$ can be taken. Arguing along the same lines as above, and making use of the uniformity of the C-minimizing property with respect to $\tau$, one concludes that the graphs $u_{\tau,i}$ have a smooth subsequential limit which satisfies the C-minimizing property. This again precludes the MITS 
$S_a$ meeting $\S$.)

\begin{figure}[h]
\begin{center}
\mbox{
\includegraphics[width=3.2in]{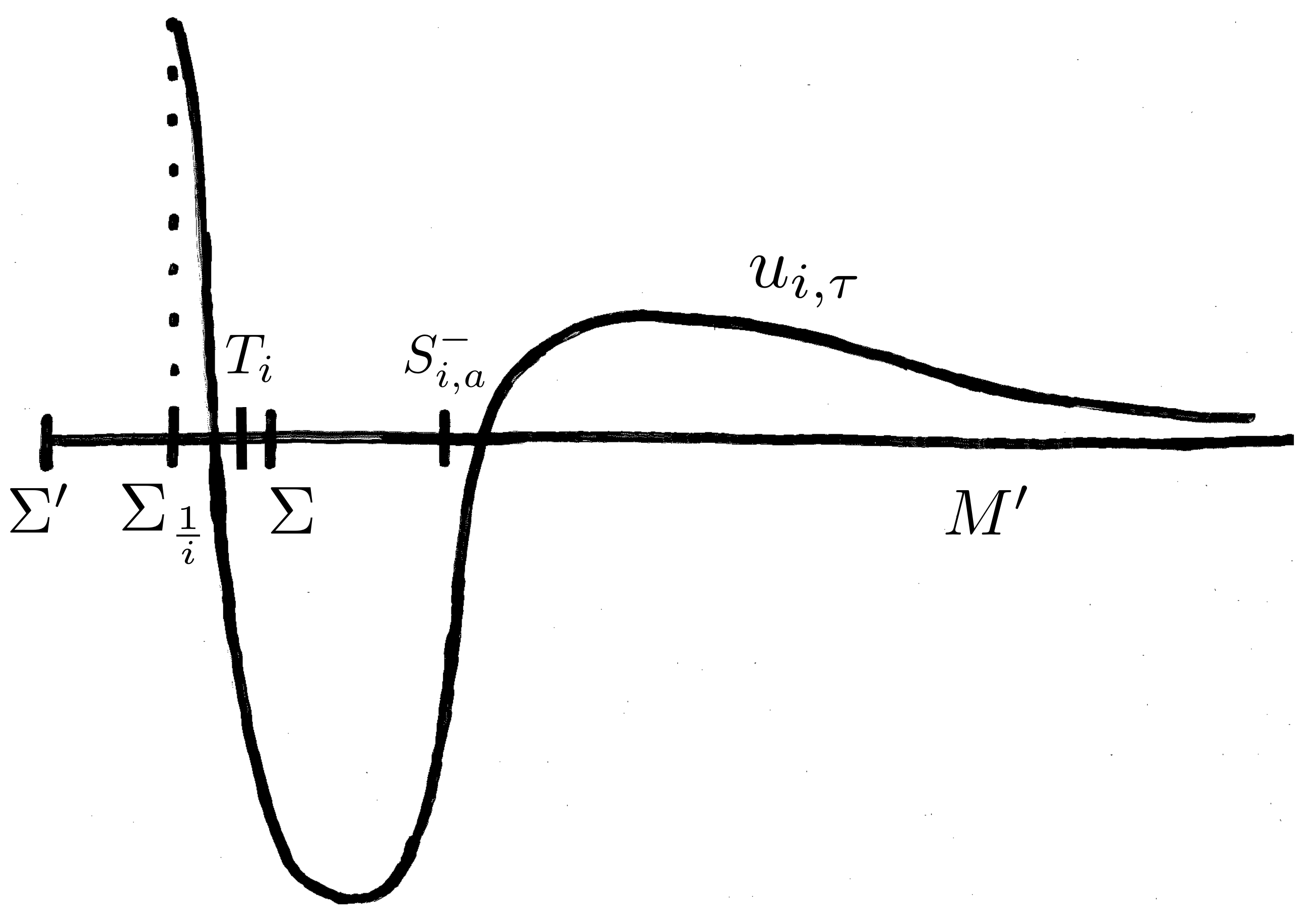}
}
\end{center}
\caption{The capillarity regularized Jang graph for $(M', g_i , K_i)$.}
\label{figure3}
\end{figure}

%

It follows from the above that for all sufficiently large $i$, $S_i$
must have only one component which is a MOTS surrounding $\Sigma'$. In
fact, by the conclusions of the compactness results for MOTS in \cite{Eichmair2} we know that, 
for all $i$ large, $S_i$ must
be a graph over $\Sigma$. To prove Theorem~\ref{thmstrictDEC} we set
$\hat M = \overline{\Omega_i^{ext}}$ and $\hat \S = S_i$ for $i$
sufficiently large. It remains to prove the regularity claimed in
point (\ref{point:jang2}). This follows by writing Jang's equation in
the form
\[
\Delta^g u 
= (1 + \nabla^k u \nabla_k u)^{1/2} g^{ij} K_{ij} 
+ (1 + \nabla^k u \nabla_k u)^{-1} \nabla^i u \nabla^j u \nabla_i \nabla_j u 
- \nabla^i u \nabla^j u K_{ij} , 
\] 
making use of the fact that the barrier argument used in constructing the 
solution to Jang's equation yields $u = O(|x|^{-\beta})$ for some 
$\beta > 0$ and using elliptic estimates. This completes the proof of 
Theorem~\ref{thmstrictDEC}.

\section{Proof of Theorem~\ref{mainthm}}

In this section we prove our main theorem. 
Assume that $(M,g,K)$ is an initial data set of dimension $n$ as in 
Theorem~\ref{mainthm}. For technical reasons we first restrict to the 
case $4 \leq n \leq 7$. The extension of the proof to the case $n=3$ is 
discussed in Remark~\ref{rem:n=3} below.

The proof is broken up into a number of steps. 

\smallskip
{\bf Step 1: Apply Theorem~\ref{thmstrictDEC} to deform to strict DEC.} 
By Theorem~\ref{thmstrictDEC}, $(M,g,K)$ may be deformed to an
initial data set satisfying the dominant energy condition with strict
inequality, $E\coloneqq \mu - |J| > 0$, while preserving the MOTS boundary,
and so that there is a solution $u$ to Jang's equation which blows up
at the MOTS boundary $\Sigma \coloneqq \partial M$ and has no further
blow-up. We denote the deformed data set again by $(M,g,K)$.

Let $\hat{M} \subset M \times \bbR$ be the graph of $u$, and let 
$\hat{g}$ be the induced Riemannian metric on $\hat{M}$. Then 
$(\hat{M}, \hat{g})$ is asymptotically flat and near $\Sigma$ it is 
asymptotic to the cylinder $(\Sigma \times \bbR, h + dt^2)$. From the 
Schoen-Yau identity 
\cite[Section~3.6]{AEM}
it follows that 
\begin{equation} \label{SYa}
\int_{\hat{M}} 
\left( 2 |\nabla \phi|_{\hat{g}}^2 + R^{\hat{g}} \phi^2 \right) 
d\mu^{\hat{g}}
\geq
\int_{\hat{M}} 2E \phi^2 d\mu^{\hat{g}}
\end{equation}
for every compactly supported smooth function $\phi$ on $\hat{M}$. 

\smallskip
{\bf Step 2: Deform the metric to exactly cylindrical ends.}
Near $\Sigma$ the Jang graph $(\hat{M},\hat{g})$ is asymptotic to the 
cylinder $(\Sigma \times \bbR, h + dt^2)$, where $h = g|_{\Sigma}$. 

We can write $\hat{M}$ as $\hat{M}_0 \cup \hat{M}_{\rm cyl}$ where 
$\hat{M}_{\rm cyl} = \Sigma \times [t_0,\infty)$. Using the normal 
exponential map of $\Sigma \times [t_0,\infty)$ in $M \times \bbR$ the 
asymptotically cylindrical end $(\hat{M}_{\rm cyl},\hat{g})$ can be
written as a graph of a function $U: \Sigma \times [t_0,\infty) \to \bbR$.
In \cite[Corollary~2]{SY2} it is proven that for every $\epsilon > 0$
there is a $t_{\epsilon} \geq t_0$ so that 
\[
|U(p,t)| + |\nabla^h U(p,t)| + |(\nabla^h)^2 U(p,t)| \leq \epsilon 
\]
for $p \in \Sigma$ and $t \geq t_{\epsilon}$. The result we refer to is 
stated in dimension $3$, but its proof holds in all dimensions.

By deforming the function $U$ to be identically zero for large $t$ we
can replace $\hat{g}$ by a metric, which we still denote by $\hat{g}$,
such that $\hat{g}= h + dt^2$ on $\Sigma \times [t_1,\infty)$, for
$t_1 > t_0$. Under the deformation of $U$ the inequality \eqref{SYa}
is almost preserved, so we get, for $t_0$ sufficiently large,
\begin{equation} \label{SYb}
\int_{\hat{M}} 
\left( (2+\epsilon) |\nabla \phi|_{\hat{g}}^2 + R^{\hat{g}} \phi^2 \right) 
d\mu^{\hat{g}}
\geq
\int_{\hat{M}} E \phi^2 d\mu^{\hat{g}}
\end{equation}
for some small $\epsilon > 0$ and all smooth compactly supported 
functions $\phi$.

{\em In each of the remaining steps the metric $\hat{g}$ is replaced
by a modified metric $\check{g}$, which is then renamed as $\hat{g}$.}

\smallskip
{\bf Step 3: Deform the metric to be flat on the asymptotically flat end.}
Working in the asymptotically flat end of $(\hat{M}, \hat{g})$, let $h$ 
denote the difference between $\hat{g}$ and the flat background metric 
$\delta$. By construction, the metric on the Jang graph is 
$\hat{g} = g + du^2$, where $u$ is the solution of Jang's equation which 
in our situation satisfies \eqref{eq:jang-decay}. 
It follows that if we write $\hat{g} = \delta + h$, then 
$h \in W^{k,p}_{-q}$. 

Let $\chi : \bbR \to \bbR $ be a smooth cut-off function such that 
$\chi(t) = 0$ for $t \leq 1$, $\chi(t) = 1$ for $t \geq 2$. For 
$(M,g)$ asymptotically flat we define $\chi_{\rho}(x) = \chi(r(x)/\rho)$ 
where $r$ is the Euclidean radial coordinate. This is defined for 
$\rho$ large and $x$ sufficiently far out in the asymptotically flat 
end, and then extended to all of $M$ by $\chi_{\rho} = 0$ inside the end. 
We have that $\nabla \chi_{\rho}$ is supported in the annulus 
$A_\rho \coloneqq \{ x \in M \mid \rho \leq r(x) \leq 2\rho \}$. Let 
$\One_{A_\rho}$ denote the characteristic function of $A_\rho$.

For $\rho$ large let 
\[
\check{g} \coloneqq \delta + (1-\chi_\rho) h = \hat{g} - \chi_\rho h .
\]
Then we have $\check{g} = \hat{g}$ for $r < \rho$ and 
$\check{g} = \delta$ for $r > 2\rho$, while in $A_\rho$ we have the 
estimates 
\[
\|\partial^m \check{g} \|_{C^0(A_\rho)} \leq C \rho^{-(n-2)/2-m}, 
\quad m=1, \dots, k-1 .
\]
In particular, we have that
\[
\| R^{\check{g}} \|_{C^0(A_\rho)} \leq C \rho^{-(n-2)/2-2} 
\]
and hence 
\[
\left| 
R^{\check{g}} - R^{\hat{g}} \frac{\sqrt{\hat{g}}}{\sqrt{\check{g}}} 
\right| 
\leq C \One_{\{ r \geq \rho\}} r^{-(n-2)/2-2}
\]
where $\frac{\sqrt{\hat{g}}}{\sqrt{\check{g}}}$ is the function such that
$d\mu^{\hat{g}} = \frac{\sqrt{\hat{g}}}{\sqrt{\check{g}}} d\mu^{\check{g}}$. 
From \eqref{SYb} we have
\begin{equation} \label{eq:Yfirst} 
\begin{split}
\int_{\hat M} 
\left( a_n |\nabla \phi|^2_{\check{g}} + R^{\check{g}} \phi^2 \right)
d\mu^{\check{g}} 
&\geq 
\int_{\hat M} \left(a_n - (2+\epsilon) \right) |\nabla\phi|^2_{\check{g}} 
d\mu^{\check{g}} \\
&\qquad 
+ (2+\epsilon) \int_{\hat M} 
\left( 
|\nabla \phi|^2_{\check{g}} 
- |\nabla \phi|^2_{\hat{g}} \frac{\sqrt{\hat{g}}}{\sqrt{\check{g}}}
\right) d\mu^{\check{g}} \\
&\qquad 
+ \int_{\hat M} 
\left( 
R^{\check{g}} - R^{\hat{g}} \frac{\sqrt{\hat{g}}}{\sqrt{\check{g}}}
\right) \phi^2
d\mu^{\check{g}} \\ 
&\qquad 
+ \int_{\hat M} E \frac{\sqrt{\hat{g}}}{\sqrt{\check{g}}}
\phi^2 d\mu^{\check{g}}  \, ,
\end{split}
\end{equation} 
where 
\[
a_n \coloneqq \frac{4(n-1)}{n-2} 
\]
is chosen so that the left hand side of \eqref{eq:Yfirst} is conformally 
invariant. The integrand in the second term on the right hand side of 
\eqref{eq:Yfirst} can be estimated in terms of 
$C \rho^{-(n-2)/2}|\nabla \phi|^2_{\check{g}}$. Using $a_n > 4$ we get
\begin{equation} \label{eq:Ysecond}
\int_{\hat M} 
\left(
a_n |\nabla \phi|^2_{\check{g}} + R^{\check{g}} \phi^2 
\right) d\mu^{\check{g}}
\geq
\int_{\hat M} 
\left( |\nabla\phi|^2_{\check{g}} + \hat{E} \phi^2 \right) 
d\mu^{\check{g}}
\end{equation}
for $\rho$ large enough, where 
\[
\hat{E} \coloneqq 
 E \frac{\sqrt{\hat{g}}}{\sqrt{\check{g}}}
+ 
\left( 
R^{\check{g}} - R^{\hat{g}} \frac{\sqrt{\hat{g}}}{\sqrt{\check{g}}}
\right).
\]
We note that $\hat{E}$ may be negative for $r \geq \rho$ due to the 
contribution from the difference of the scalar curvatures. By construction, 
we have
\[
E \leq Cr^{-n}
\]
for $r \geq R_0$, and 
\begin{equation} \label{eq:hatEbound}
|\hat{E}| \leq Cr^{-(n-2)/2-2}
\end{equation}
for some constant $C$. 

Next we will estimate the right hand side of \eqref{eq:Ysecond}
using a weighted Hardy inequality. Let $\chi_\rho$ be the smooth cut-off 
function on $\bbR^n$ defined by $\chi_\rho(x) \coloneqq \chi(|x|/\rho)$ 
where $\chi$ is the function introduced above. The following lemma
follows by the standard Hardy inequality (see for example Section~2.1.6.
of \cite{Mazja}) applied to $\chi_{\rho} u$, followed by an application
of the Cauchy-Schwartz inequality.
\begin{lem} \label{lem:Hardy-cutoff} 
There is a constant $C_n > 0$ depending only on $n$, so that 
\[
C_n \int_{\bbR^n} u^2 \chi_\rho^2 r^{-2} d\mu 
\leq 
\int_{\bbR^n} \left( 
\chi_\rho^2 |\nabla u|^2_\delta + \frac{1}{\rho^2} \One_{A_\rho} u^2 
\right) d\mu
\]
for all $u \in C^\infty_0(\bbR^n)$. 
\end{lem}
If $(M,g)$ is asymptotically flat we have
\[
\int_M \chi^2_\rho |\nabla u|^2_g d\mu^g 
\geq 
C \int_{\bbR^n} \chi^2_\rho |\nabla u|^2_\delta d\mu 
\]
for $\rho$ sufficiently large from the identification of the end with
$\bbR^n \setminus B$. Applying Lemma~\ref{lem:Hardy-cutoff} gives the
following corollary.
\begin{cor} \label{cor:Hardy-cutoff} 
Let $(M,g)$ be an asymptotically flat Riemannian manifold of dimension
$n$, as in Definition~\ref{def:AF} (with $K=0$). For $R_0$
sufficiently large and $\rho > R_0$, there is a constant $C_n > 0$
depending only on $n$ such that
\[
C_n \int_M \chi_\rho^2 r^{-2} u^2 d\mu^g 
\leq \int_M 
\left( |\nabla u|_g^2 + \frac{1}{\rho^2} \One_{A_{\rho}} u^2 \right) 
d\mu^{g}
\]
for all $u \in C^\infty_0(M)$.
\end{cor} 

Fix some large $\rho_0$. Corollary~\ref{cor:Hardy-cutoff} with 
$\rho = \rho_0$ gives us the estimate 
\[
\int_{\hat M} 
\left(
\lambda |\nabla \phi|_{\check{g}}^2 + \hat{E} \phi^2 
\right) d\mu^{\check{g}}
\geq 
\int_{\hat M} \left( 
\lambda C_n \chi_{\rho_0}^2 r^{-2} + \hat{E} - \lambda \rho_0^{-2} \One_{A_{\rho_0}} 
\right) \phi^2 d\mu^{\check{g}}
\] 
for $\lambda > 0$. With $\lambda = \rho^{-1/4}$, where $\rho$ is the
parameter in the definition of $\check{g}$, the inequality
\[
\lambda C_n \chi^2_{\rho_0} r^{-2} + \hat{E}
= \rho^{-1/4} C_n \chi^2_{\rho_0} r^{-2} + \hat{E} 
> 0 
\] 
holds trivially on the region inside $r = \rho$ (since $\check{g} = \hat{g}$ there)
and by \eqref{eq:hatEbound} it holds on the region $r \ge \rho$ if we choose $\rho$ 
sufficiently large. Similarly, keeping $\rho_0$ fixed and choosing $\rho$ sufficiently 
large, we have
\[
\hat{E} - \lambda \rho_0^{-2} \One_{A_{\rho_0}} =
\hat{E} - \rho^{-1/4} \rho_0^{-2} \One_{A_{\rho_0}} > 0
\]
in $A_{\rho_0}$, since then $\hat{E} = E$ on $A_{\rho_0}$ so $\hat{E}$ is positive and  
independent of $\rho$ there.

Redefining $\hat{E}$ as 
\[
\lambda C_n \chi_{\rho_0}^2 r^{-2} + \hat{E} - \lambda \rho_0^{-2} \One_{A_{\rho_0}} 
= \rho^{-1/4} C_n \chi_{\rho_0}^2 r^{-2} + \hat{E} - \rho^{-1/4} \rho_0^{-2} \One_{A_{\rho_0}}
\] 
we have that $\hat{E} > 0$ by the above choices. Redefining 
$\hat{g}$ as $\check{g}$ we get from \eqref{eq:Ysecond} that
\begin{equation}\label{eq:SYc}
\int_{\hat M} 
\left( a_n |\nabla \phi|^2_{\hat{g}} + R^{\hat{g}} \phi^2 \right) 
d\mu^{\hat{g}} 
\geq \int \hat{E} \phi^2 d\mu^{\hat{g}}
\end{equation}
for all compactly supported smooth functions $\phi$. By construction 
we now have that $\hat{g}$ is flat on the asymptotically flat end and
\[
\hat{E} \geq C r^{-2}
\]
for $r$ large.

{\bf Step 4: Conformal compactification of the asymptotically flat end.} 
The flat background metric $\delta$ compactifies to the standard round 
metric on $S^n$ by the conformal change
\[
\left( \frac{2}{1+r^2} \right)^{2} \delta 
= \left( \frac{2}{1+r^2} \right)^{2} (dr^2 + r^2 \sigma)
= d\theta^2 + \sin^2 \theta \sigma
\]
where $2r/(1+r^2) = \sin \theta$ and $\sigma$ is the round metric on 
$S^{n-1}$. Define 
\[
\alpha \coloneqq \left( \frac{2}{1+r^2} \right)^{(n-2)/2} 
\]
in coordinates on the asymptotically flat end and extend $\alpha$ to a
positive function on all of $\hat{M}$ with $\alpha=1$ on the 
cylindrical end. Set
$\check{g} = \alpha^{4/(n-2)} \hat{g}$ and let $\check{M}$ be $\hat{M}$ 
with the asymptotically flat end compactified by adding a point
$P_\infty$ at infinity. Then $(\check{M}, \check{g})$ is isometric to 
the standard round metric on $S^n$ in a neighbourhood of the new point 
at infinity.

The conformal Laplacians
$L^{\hat{g}} = a_n \Delta^{\hat{g}} + R^{\hat{g}}$ and 
$L^{\check{g}} = a_n \Delta^{\check{g}} + R^{\check{g}}$ are 
related by
\[
L^{\hat{g}} \phi = \alpha^{\frac{n+2}{n-2}} L^{\check{g}} ( \alpha^{-1} \phi ).
\]
Further, we have 
\[
d\mu^{\hat{g}} = \alpha^{-\frac{2n}{n-2}} d\mu^{\check{g}}.
\]
These equations, together with inequality \eqref{eq:SYc} gives us
\[\begin{split} 
\int_{\hat{M}} 
(\alpha^{-1} \phi) L^{\check{g}} (\alpha^{-1} \phi )
d\mu^{\check{g}} 
&=
\int_{\hat{M}} \phi L^{\hat{g}} \phi d\mu^{\hat{g}} \\
&=
\int_{\hat{M}} 
\left( a_n |\nabla \phi|_{\hat{g}}^2 + R^{\hat{g}} \phi^2 \right) 
d\mu^{\hat{g}} \\
&\geq
\int_{\hat{M}} \hat{E} \phi^2 d\mu^{\hat{g}} \\
&=
\int_{\hat{M}} \hat{E} \alpha^{-4/(n-2)}
(\alpha^{-1} \phi)^2 d\mu^{\check{g}}.
\end{split}\]
So with $\check{E} \coloneqq \hat{E} \alpha^{-4/(n-2)}$ it holds that 
\[
\begin{split}
\int_{\check{M}} 
\left( a_n |\nabla \phi|_{\check{g}}^2 + R^{\check{g}} \phi^2 \right) 
d\mu^{\check{g}}
&\geq
\int_{\check{M}} \hat{E} \alpha^{-4/(n-2)} \phi^2 d\mu^{\check{g}} \\
&=
\int_{\check{M}} \check{E} \phi^2 d\mu^{\check{g}}
\end{split}
\]
for all smooth compactly supported functions $\phi$ on $\check{M}$
whose support does not contain $P_\infty$.

In terms of the spherical radial coordinate $\theta$ at $P_\infty$, we 
have that $\hat{E} = O(\theta^2)$ and $\alpha^{-4/(n-2)} = O(\theta^{-4})$. 
This means that $\check{E} = O(\theta^{-2})$ which is compatible with the 
fact that we made use of the Hardy inequality in the construction of 
$\hat{E}$. We can now modify $\check{E}$ by decreasing its values in a 
neighborhood of the point at infinity, and thereby replace it by a 
bounded smooth function which is uniformly positive on $\check{M}$.
We finally get the inequality 
\begin{equation} \label{eq:SYd}
\int_{\check{M}} 
\left( a_n |\nabla \phi|_{\check{g}}^2 + R^{\check{g}} \phi^2 \right) 
d\mu^{\check{g}}
\geq 
\int_{\check{M}} \check{E} \phi^2 d\mu^{\check{g}}
\end{equation} 
for $\phi \in C^\infty_0(\check{M} \setminus P_\infty)$. 
A cut-off function argument shows that \eqref{eq:SYd} is valid for 
$\phi \in C^\infty_0(\check{M})$.

We redefine $(\hat{M},\hat{g})$ and $\hat{E}$ as 
$(\check{M}, \check{g})$ and $\check{E}$. This is then a metric with the 
asymptotically flat end compactified by a point, and an exact cylindrical 
end, such that
\begin{equation} \label{eq:SYe}
\int_{\hat{M}} 
\left( a_n |\nabla \phi|_{\hat{g}}^2 + R^{\hat{g}} \phi^2 \right) 
d\mu^{\hat{g}}
\geq 
\int_{\hat{M}} \hat{E} \phi^2 d\mu^{\hat{g}}
\end{equation} 
for $\phi \in C^\infty_0(\hat{M})$. 
 
\smallskip
{\bf Step 5: Conformal change to positive scalar curvature on the 
cylindrical end.} The next two steps are motivated by results in \cite{CM14}.

The estimate \eqref{eq:SYe} holds in particular for functions of compact 
support on the cylindrical end $(\Sigma \times [0,\infty),h+dt^2)$ 
of $(\hat{M},\hat{g})$. On the cylinder we have that $\hat{E}$ is larger 
than a constant $C$, so
\begin{equation} \label{SYf}
\int_{\Sigma \times \bbR} 
\left( a_n |\nabla \phi|_{h+dt^2}^2 + R^{h} \phi^2 \right) 
d\mu^{h+dt^2}
\geq
C \int_{\Sigma \times \bbR} \phi^2 d\mu^{h + dt^2}
\end{equation}
for functions $\phi$ with compact support on $\Sigma \times \bbR$. 
Let $\chi(t)$ be smooth compactly supported function on $\bbR$ with 
\[
\int_{\bbR} (\chi(t))^2 \,dt = 1, 
\qquad
a_n \int_{\bbR} (\chi'(t))^2 \,dt \leq \frac{C}{2},
\]
and let $v$ be any smooth function on $\Sigma$. If we set $\phi = \chi v$
in \eqref{SYf} we get
\[ \begin{split}
&\int_{\Sigma} \int_{\bbR}
\left( a_n (\chi'(t))^2 v^2 +
a_n (\chi(t))^2|\nabla v|_{h}^2 + R^{h} (\chi(t))^2 v^2 \right) 
\, dtd\mu^{h} \\
&\quad \geq
C \int_{\Sigma}\int_{\bbR} (\chi(t))^2 v^2 \, dt d\mu^{h}, 
\end{split} \]
which by the properties of $\chi(t)$ gives us 
\[
\int_{\Sigma} 
\left( a_n |\nabla v|_{h}^2 + R^{h} v^2 \right) 
d\mu^{h}
\geq
\frac{C}{2} \int_{\Sigma} v^2 d\mu^{h}.
\]
This means that the operator 
${\mathcal L}^h \coloneqq - a_n \Delta^h + R^{h}$ on $\Sigma$
has a spectrum consisting only of positive eigenvalues. 

Let $v_0$ be a positive eigenfunction corresponding to the smallest 
eigenvalue $\mu_0$ of ${\mathcal L}^h$, that is 
${\mathcal L}^h v_0 = \mu_0 v_0$. Let $v$ be a positive function on 
$\hat{M}$ which is equal to $v_0$ on $\Sigma \times [0,\infty)$ and set
\[
\check{g} \coloneqq v^{4/(n-2)} \hat{g}.
\]
Then $R^{\check{g}} = v^{-\frac{n+2}{n-2}} L^{\hat{g}} v$, so on 
$\Sigma \times [0,\infty)$ we have
\[
R^{\check{g}} = v_0^{-\frac{n+2}{n-2}} L^{\hat{g}} v_0
= v_0^{-\frac{n+2}{n-2}} {\mathcal L}^h v_0
= \mu_0 v_0^{-4/(n-2)} \geq c > 0.
\]
The estimate \eqref{eq:SYe} and conformal invariance gives us 
\[ 
\int_{\hat{M}} 
\left( a_n |\nabla \phi|_{\check{g}}^2 + R^{\check{g}} \phi^2 \right) 
d\mu^{\check{g}}
\geq
\int_{\hat{M}} \hat{E} v^{-4/(n-2)} \phi^2 d\mu^{\check{g}}. 
\]
We redefine $\hat{g}$ as $\check{g}$ and $\hat{E}$ as $\hat{E} v^{-4/(n-2)}$.
Then $(\hat{M},\hat{g})$ has positive scalar curvature on the cylindrical
end, it satisfies all the properties from step 4, and 
\begin{equation} \label{eq:SYh}
\int_{\hat{M}} 
\left( a_n |\nabla \phi|_{\hat{g}}^2 + R^{\hat{g}} \phi^2 \right) 
d\mu^{\hat{g}}
\geq
\int_{\hat{M}} \hat{E} \phi^2 d\mu^{\hat{g}} 
\end{equation}
for $\phi \in C^\infty_0(\hat{M})$.

\smallskip
{\bf Step 6: Conformal change to positive scalar curvature everywhere.}
We follow the argument in the proof of Proposition~4.6 in \cite{CM14}. 

We first prove that the conclusions of Lemma~4.5 in \cite{CM14} 
follow from \eqref{eq:SYh}. The first of these conclusions is that the 
$L^2$-spectrum of 
$L^{\hat{g}} = - a_n \Delta^{\hat{g}} + R^{\hat{g}}$ is 
contained in $[0,\infty)$, which clearly follows from \eqref{eq:SYh}. 

The second conclusion is that $L^{\hat{g}}$ does not have zero as an
eigenvalue. Assume that $u$ is an $L^2$ function with $L^{\hat{g}}u = 0$. 
Since $R^{\hat{g}} \geq c > 0$ on the cylindrical end of 
$(\hat{M},\hat{g})$ the ``tangential part'' of the operator 
$L^{\hat{g}}$ on the cylindrical end has a spectrum consisting only of 
positive eigenvalues. Separation of variables tells us that the $L^2$ 
function $u$ must have exponential decay on the cylindrical end. Since 
$u$ has exponential decay we can multiply with cut-off functions, insert 
in \eqref{eq:SYh} and integrate by parts to conclude that $u=0$, 
since $\hat{E}$ is strictly positive. 

Let $\theta$ be a positive function which is equal to 
$R^{\hat{g}} $ outside a large compact set. We want to solve the 
equation
\[
L^{\hat{g}} f = \theta
\]
for a positive function $f$ which tends to $1$ on the cylindrical end.
Set $f = 1+\alpha$, so that 
\[
L^{\hat{g}} \alpha = L^{\hat{g}} (f-1) 
= \theta - R^{\hat{g}} =: \tilde{\theta}
\]
where $\tilde{\theta}$ has compact support. In the proof of
Proposition~4.6 in \cite{CM14} there is an argument using barrier
functions to show that there is a solution $\alpha$, such that 
$f = 1+\alpha > 0$.

Making a conformal change with $f$ we get the metric 
$f^{4/(n-2)} \hat{g}$ which has scalar curvature 
\[
R^{f^{4/(n-2)} \hat{g}} = f^{-\frac{n+2}{n-2}} L^{\hat{g}} f
= f^{-\frac{n+2}{n-2}} \theta > 0.
\]
This metric does not have an exact cylindrical end, but since 
$\alpha$ decays exponentially we can deform it to be zero outside a large 
compact set, so that $f = 1$ outside this compact set, while preserving 
positivity of scalar curvature.
Finally, we set $\check{g} = f^{4/(n-2)} \hat{g}$ with the modified 
function $f$, and we cut off the cylindrical end of $\hat{M}$ to get a 
manifold $(\check{M},\check{g})$ with boundary satisfying all the stated 
properties. 

This completes the proof of Theorem~\ref{mainthm} for the case 
$4 \leq n \leq 7$. The following remark deals with the case $n=3$. 

\begin{rem}\label{rem:n=3}
In the case $n=3$, the decay of the initial data $(g_i, K_i)$ provided by the
density theorem \cite[Theorem~22]{EHLS} is not compatible with solving
Jang's equation, due to the slow decay of the mean curvature 
$\tr^{g_i} {K_i}$. A cut-off argument similar to that used in
Step 3 above can be used to modify this data near infinity to get
$\tr^{g_i}{K_i} = O(|x|^{-\gamma})$ for $\gamma > 2$. The
modified data, however, fails to have strict DEC in a neighborhood of
infinity. Performing the cut-off at a sufficiently large radius, one
finds upon constructing a solution to Jang's equation for this modified
data, that the Hardy inequality argument used in Step 3 can be applied
again to recover the inequality \eqref{eq:SYc}. This approach allows
us to extend the result of Theorem~\ref{mainthm} to the case $n=3$. We
leave the details to the reader.
\end{rem}

\section{Obstructions to positive scalar curvature} 
\label{sec_PSC}

In this section we will discuss conclusions about the topology of manifold 
$M$ which can be drawn from the Theorem~\ref{mainthm}.

\subsection{Three dimensions}

First we consider the case when the dimension $n=3$. Assume that $M$
is a connected oriented $3$-manifold, with connected boundary,
satisfying the conclusion of Theorem~\ref{mainthm}. Then its boundary
$\d M$ is diffeomorphic to a $2$-sphere. From
\cite[Theorem~3.4]{RosenbergStolz01} we know that the positive scalar
curvature metric $h$ on $\partial M$ is isotopic to the standard
metric on $S^2$.
Using the isotopy of metrics we get a positive scalar curvature metric
on the cylinder $S^2 \times I$ if the interval $I$ is long enough (so that the 
isotopy is run through very slowly). Capping off the cylinder with the standard 
metric on the hemisphere we get a positive scalar curvature metric on 
the 3-ball $B$ which has $h$ as its induced metric on the boundary and is 
product near the boundary. By a classical result of Gromov
and Lawson \cite{GL}, and the positive resolution of the Poincar\'e
conjecture, $\tilde M = \check{M} \cup_{\partial M} B$ must be a
connected sum of spherical space forms (manifolds of the form
$S^3/\Gamma$, where $\Gamma$ is a finite group of isometries of $S^3$) 
and copies of $S^2 \times S^1$. We thus have the following.

\begin{prop}\label{3dtopo}
Let $(M,g,K)$ be a $3$-dimensional initial data set satisfying the
hypotheses of Theorem~\ref{mainthm}, and assume $M$ is
orientable. Then $M$ is diffeomorphic to $\bbR^3 \# N\setminus B$,
where $N$ is a connected sum (possibly empty) of spherical space forms
and copies of $S^2 \times S^1$, and $B$ is an open Euclidean ball.
\end{prop}

A key assumption in Proposition \ref{3dtopo} is that there are no
MOTS/MITS in $M \setminus \d M$. In \cite{EGP13} it was shown that $M
\approx \bbR^3 \setminus B$, under the stronger assumption that there are
no {\it immersed} MOTS (as defined in \cite{EGP13}) in
$M \setminus \d M$. It remains an interesting open question whether
the same conclusion can be reached under the assumption of no MOTS/MITS.

\subsection{Index obstructions}

In general dimensions there are obstructions to the existence of
positive scalar curvature metrics coming from the index of 
Dirac operators on
the manifold. The setting which is relevant here is with a compact
spin manifold $M$ with boundary $\partial M$ and a given metric of
positive scalar curvature $h$ defined on the boundary. Actually, the
metric on the boundary is only required to have invertible Dirac 
operator, which holds also when $h$ is conformal to a positive scalar
curvature metric. The metric $h$
is then extended to a metric $g$ on $M$ with the only requirement that
$g$ is a product $h + dt^2$ in a neighbourhood of the boundary. Using
$g$ a Dirac operator on $M$ is defined. The index
$\operatorname{ind}(M,h)$ of this Dirac operator defines an element of
$K_{n}(C^*_r \pi_1(M))$, that is of the $K$-theory of the reduced
$C^*$-algebra $C^*_r \pi_1(M)$ of the fundamental group $\pi_1(M)$,
see for example \cite{Schick14} or \cite{PiazzaSchick14}. The index
depends only on the pair $(M,h)$ up to cobordism, meaning that if
$(M',h')$ is cobordant to $(M,h)$ through a manifold with corners, and
the induced cobordism from $\partial M$ to $\partial M'$ is equipped
with a positive scalar curvature metric which restricts to $h$,
resp. $h'$, then $\operatorname{ind}(M',h') =
\operatorname{ind}(M,h)$. Again, it is actually only required that
the induced bordism between the boundaries has invertible Dirac 
operator. If the metric $h$ can be extended to a
metric on $M$ which has positive scalar curvature and is a product
near the boundary, then $\operatorname{ind}(M,h) = 0$ by the
Schr\"odinger-Lichnerowicz formula.

Applied to our setting we thus have the index 
$\operatorname{ind}(\check{M},h) \in K_{n}(C^*_r \pi_1(M))$ defined 
for the compactified exterior $\check{M}$ of any strictly stable MOTS
in $M$, since by Theorem~\ref{GS-thm} the induced metric on the MOTS
boundary is conformal to a metric of positive scalar curvature. In
case the MOTS is outermost and the initial data set is spin and
satisfies the dominant energy condition we conclude from
Theorem~\ref{mainthm} that $\operatorname{ind}(\check{M},h)$ vanishes.
The conclusion is actually that the index with respect to the boundary 
metric of $\check{g}$ vanishes, but since this is conformal to a small
perturbation of $h$ the index is the same for the two boundary metrics.

\begin{prop}\label{indexobstr}
Let $(M,g,K)$ be an $n$-dimensional, $3 \leq n \leq 7$, initial data
set satisfying the hypotheses of Theorem~\ref{mainthm}. Assume that
$M$ is spin. Then $\operatorname{ind}(\check{M},h) = 0$ where
$h$ is the induced metric on $\d M$, and $\check{M}$ is the
exterior of this MOTS with the asymptotically flat end compactified by
a point.
\end{prop}

As a simple application, we have the following. 

\begin{prop}\label{sum}
Let $(M,g,K)$ be an initial data set satisfying the hypotheses of
Theorem~\ref{mainthm}, and assume $M$ is spin. Let $M'$ be a manifold
which can be expressed as a connected sum, $M' = M \# X$, where $X$ is
a closed spin manifold. Suppose there exists initial data on $M'$
satisfying the dominant energy condition and coinciding with $(g,K)$
near the boundary. Then either $\operatorname{ind}(X) = 0$ or there
exists a MOTS in the exterior of the boundary of $M'$.
\end{prop}

\begin{proof}
We have that $\check{M'} = \check{M} \# X$, where $\check{M}$ is the
one point compactification of $M$. The metrics on the boundaries from
Theorem~\ref{mainthm} are both conformal to a small perturbation of
the boundary metric $h$. The proposition then follows from
Proposition~\ref{indexobstr}, together with the cobordism invariance of
the index and the fact that a connected sum is cobordant to the
disjoint union of the summands.
\end{proof}

For example, if, in the context of the Proposition \ref{sum}, $M$ is a
four manifold and $X$ is a K3 surface, then there must exist a MOTS in
the exterior of $M'$.

We now make use of Gromov and Lawson's notion of enlargeability, as
described in \cite{GL} and references therein.

\begin{prop}\label{sum2}
Let $(M,g,K)$ be an initial data set satisfying the hypotheses of
Theorem \ref{mainthm}, and assume $M$ is spin. Suppose $M$ can be
expressed as, $M = N \# X$, where $X$ is a closed manifold. Then $X$
is not enlargeable.
\end{prop}

\begin{proof}
Let $\check{P}$ denote the double of $\check{M} = \check{N} \# X$.
$\check{P}$ is the connected sum of a closed spin manifold and an
enlargeable manifold (namely $X$ or its copy in the double). But this
contradicts the fact, which is an immediate consequence of Theorem
\ref{mainthm}, that $\check{P}$ admits a metric of positive scalar
curvature.
\end{proof}

So, for example, as follows from results in \cite{GL}, if $(M,g,K)$ is
as in Proposition~\ref{sum2}, then $X$ could not be a torus, or, more
generally, a manifold that admits a metric of nonpositive sectional
curvature.

This suggests a relationship between the results we have established
here and a natural problem for the existence of metrics of positive
scalar curvature. 
\begin{ques} \label{question1}
Let $(M,g,K)$ be an initial data set satisfying the hypotheses of
Theorem \ref{mainthm}, with $\partial M$ connected and spherical. Does
the closed manifold $\bar{M}$, obtained by compactifying the end of
$M$ and attaching a standard ball to $\partial M$, admit a metric of
positive scalar curvature?
\end{ques}
One should compare this with Problem 6.1 in \cite{RosenbergStolz01}.
We should note however that in this context we have
\begin{prop} \label{partial-ans}
Let $(M,g,K)$ be an initial data set satisfying the hypotheses of
Theorem \ref{mainthm}, with $\partial M$ connected and spherical and
assume $M$ is spin. Then the closed manifold $\bar{M}$, obtained by
compactifying the end of $M$ and attaching a standard ball to
$\partial M$, is not enlargeable.
\end{prop}

\begin{proof}
Suppose $\bar{M}$ were enlargeable then, $\bar{M}\#\bar{M}$ would also
be enlargeable by \cite{GL}, but this is precisely the compactified
double of $\check{M}$ which we have shown admits a metric of positive
scalar curvature.
\end{proof}

\subsection{Minimal hypersurface obstructions}

For closed manifolds of dimension less than $8$ one finds obstructions
to positive scalar curvature by using the fact that area-minimizing
hypersurfaces of a manifold with positive scalar curvature also allow
positive scalar curvature metrics, see \cite{SYscalar2}. In this
context, Schoen and Yau introduce, for each $n$, a class of
$n$-manifolds $\mathcal{C}_n$, and prove that a closed orientable
manifold $M$ of dimension $n$, $3 \le n \le 7$, having positive scalar
curvature must belong to $\mathcal{C}_n$.

Similar classes of manifolds can be defined for compact manifolds $M$
with mean convex boundary, $H_{\d M} \ge 0$ (with respect to the
outward normal). Let $\mathcal{C}'_3$ be the class of compact
orientable $3$-manifolds with (possibly empty) mean convex boundary
such that for any finite covering manifold $\tilde M$ of $M$,
$\pi_1(\tilde M)$ contains no subgroup isomorphic to the fundamental
group of a compact surface of genus $\ge 1$. In general, we say that
an $n$-dimensional, $n \ge 4$, compact orientable manifold $M$ with
(possible empty) mean convex boundary is of class $\mathcal{C}'_n$ if
for any finite covering space of $M$, every nontrivial codimension one
homology class can be represented by an embedded compact hypersurface
of class $\mathcal{C}'_{n-1}$. Then, using the results in \cite{HS}
for compact $3$-manifolds with mean convex boundary, the proof of
Theorem 1 in \cite{SYscalar2} is easily modified to show that a
compact orientable manifold $M$ with mean convex boundary of dimension
$n$, $3 \le n \le 7$, having positive scalar curvature must belong to
$\mathcal{C}'_n$. This immediately yields the following.

\begin{prop}
Let $(M,g,K)$ be an $n$-dimensional initial data set, $3 \le n \le 7$,
satisfying the hypotheses of Theorem \ref{mainthm}, and assume $M$ is
orientable. Then the compactification $\check M$ belongs to the class
$\mathcal{C}'_n$.
\end{prop}

\bibliographystyle{amsplain}
\bibliography{geomids}

\end{document}